\documentclass[11pt, a4paper]{article}

\pdfoutput=1

\usepackage[margin=1in]{geometry}
\usepackage[utf8]{inputenc}
\usepackage[T1]{fontenc}
\usepackage[colorlinks]{hyperref}
\usepackage{microtype}
\usepackage{xcolor}

\hypersetup{
  linkcolor=[rgb]{0.3,0.3,0.6},
  citecolor=[rgb]{0.2, 0.6, 0.2},
  urlcolor=[rgb]{0.6, 0.2, 0.2}
}

\usepackage{amsthm}

\usepackage[T1]{fontenc}
\usepackage{lmodern}
\usepackage{mathtools, amsthm, amsfonts, amssymb}

\usepackage{mathrsfs}
\usepackage{pbox}

\usepackage{authblk}
\usepackage{booktabs}
\usepackage{braket}
\usepackage{commath}
\usepackage{caption}
\usepackage{titlesec}
\titlelabel{\thetitle. }
\titleformat*{\section}{\large\bfseries}

\usepackage[capitalise]{cleveref}
\usepackage{accents}

\usepackage{tikz}
\usetikzlibrary{positioning}

\usepackage{chngcntr}
\usepackage{apptools}
\AtAppendix{\counterwithin{theorem}{section}}

\newtheorem{theorem}{Theorem}[section]
\newtheorem*{theorem*}{Theorem}
\newtheorem{proposition}[theorem]{Proposition}
\newtheorem{lemma}[theorem]{Lemma}

\newtheorem{definition}[theorem]{Definition}

\newtheorem*{problem*}{Problem}

\newtheorem{remark}[theorem]{Remark}

\DeclareMathAccent{\wtilde}{\mathord}{largesymbols}{"65}

\DeclareMathOperator{\supp}{\text{supp}}

\newcommand{\N}{\mathbb{N}}

\newcommand{\R}{\mathbb{R}}

\newcommand{\cP}{\mathcal{P}}

\newcommand{\lv}{\lvert}

\usepackage{letltxmacro}
\LetLtxMacro\orgvdots\vdots
\LetLtxMacro\orgddots\ddots

\makeatletter
\DeclareRobustCommand\vdots{%
  \mathpalette\@vdots{}%
}
\newcommand*{\@vdots}[2]{%
  \sbox0{$#1\cdotp\cdotp\cdotp\m@th$}%
  \sbox2{$#1.\m@th$}%
  \vbox{%
    \dimen@=\wd0 %
    \advance\dimen@ -3\ht2 %
    \kern.5\dimen@
    \dimen@=\wd2 %
    \advance\dimen@ -\ht2 %
    \dimen2=\wd0 %
    \advance\dimen2 -\dimen@
    \vbox to \dimen2{%
      \offinterlineskip
      \copy2 \vfill\copy2 \vfill\copy2 %
    }%
  }%
}
\DeclareRobustCommand\ddots{%
  \mathinner{%
    \mathpalette\@ddots{}%
    \mkern\thinmuskip
  }%
}
\newcommand*{\@ddots}[2]{%
  \sbox0{$#1\cdotp\cdotp\cdotp\m@th$}%
  \sbox2{$#1.\m@th$}%
  \vbox{%
    \dimen@=\wd0 %
    \advance\dimen@ -3\ht2 %
    \kern.5\dimen@
    \dimen@=\wd2 %
    \advance\dimen@ -\ht2 %
    \dimen2=\wd0 %
    \advance\dimen2 -\dimen@
    \vbox to \dimen2{%
      \offinterlineskip
      \hbox{$#1\mathpunct{.}\m@th$}%
      \vfill
      \hbox{$#1\mathpunct{\kern\wd2}\mathpunct{.}\m@th$}%
      \vfill
      \hbox{$#1\mathpunct{\kern\wd2}\mathpunct{\kern\wd2}\mathpunct{.}\m@th$}%
    }%
  }%
}

\author{Asger Kj\ae rulff Jensen}
\affil{QMATH, Department of Mathematical Sciences, University of Copenhagen, Universitetsparken 5, 2100 Copenhagen, Denmark}
\title{Asymptotic majorizaiton of finite probability distributions}

\makeatother
\begin{document}
\maketitle
%\vspace*{1em}
%\begin{center}
%\Large\textbf{The spectrum of asymptotic LOCC conversion}\par
%\vspace{1em}
%\large Asger Kjærulff Jensen, P\'eter Vrana\par
%\end{center}
%\vspace{0.5em}

\begin{abstract}
This paper studies majorization of high tensor powers of finitely supported probability distributions. Viewing probability distributions as a resource with majorization as a means of transformation corresponds to the resource theory of pure bipartite quantum states under LOCC transformations vis-à-vis Nielsen's Theorem \cite{PhysRevLett.83.436}. In \cite[Example 8.26]{fritz_2017} a formula for the asymptotic exchange rate between any two finitely supported probability distributions was conjectured. The main result of the present paper is Theorem \ref{th1}, which resolves this conjecture.
\end{abstract}
\section{Introduction}

Majorization of probability distributions is an important notion in the field of information theory. Given probability distributions $P$ and $Q$, we ask whether $P^{\otimes n}\preceq Q^{\otimes n}$ for large $n$, and we ask how large $r\in \R$ is allowed to be for $P^{\otimes n}\preceq Q^{\otimes \lfloor n r\rfloor}$ to be true for large $n$. We denote the supremum of such $r$ by $R(P,Q)$. This question is of particular interest to the author as it relates to LOCC transformation of bipartite pure quantum states. In this context $R(P,Q)$ is the optimal rate by which one can extract copies of the bipartite state with Schmidt coefficients $Q$ from copies of the bipartite state with Schmidt coefficients $P$. In Theorem \ref{th1} we show that
\begin{equation}
	R(P,Q) = \min_{\alpha\in [0,\infty]} \frac{H_\alpha(P)}{H_\alpha(Q)}.
\end{equation}
The main tool for obtaining this formula is a description of the growth exponents defined in \ref{def1}. This description is found in Proposition \ref{HayashiLemma}. The quantum information reader might note the resemblance with the well known entanglement manipulation theorem \cite[ch. 19.4]{Nielsen:2011:QCQ:1972505}, which states that the exchange rate, when one allows for non-exact LOCC transformations while demanding that fidelity goes to 1, is given by the ratio of the Shannon entropies (i.e. $\alpha=1$). In the quantum resource theory interpretation of present paper we demand that the transformation is exact.
\\\\

This paper started out as a section in \cite{2018arXiv180705130K}. As the section grew larger and we were made aware of the conjecture in \cite[Example 8.26]{fritz_2017}, it was decided to write a seperate paper proving this conjecture, while developing the necessary techniques properly.
\section{Asymptotic exponents}

Given a probability distribution $P:X\to [0,1]$ with finite support $\left|\supp(P)\right|=d$, we let $P^\downarrow:[d]=\{1,\ldots ,d\} \to [0,1]$ be $P$ ordered non-increasingly. We may naturally extend $P^\downarrow:\N\to [0,1]$ by $P(i)=0$ for $i>d$. In this paper all probability distributions will have finite support.
\begin{definition}
	Given two probability distributions $P,Q$, we say that $Q$ majorizes $P$, written $P\preceq Q$, if
	\begin{equation}
		\sum_{i=1}^N P^\downarrow(i) \le \sum_{i=1}^N Q^\downarrow(i)
	\end{equation}
	for all $N\in \N$.
\end{definition}
\begin{definition}
	Given a probability distribution $P$ with $\supp(P)=X$ and an $\alpha\in [0,\infty)\backslash \{1\}$, the Rényi $\alpha$-entropy is given by
	\begin{equation}
		H_{\alpha}(P)=\frac{1}{1-\alpha}\log \sum_{x\in X} P(x)^\alpha.
	\end{equation}
	For $\alpha\in \{1,\infty\}$, $H_\alpha(P)$ is defined by taking the limit $\lim_{\beta\to \alpha} H_\beta(P)$, that is
	\begin{equation}
		H_1(P)=H(P)=-\sum_{x\in X} P(x)\log P(x)
	\end{equation}
	\begin{equation}
		H_\infty(P) = -\max_{x\in X} \log P(x).
	\end{equation}
\end{definition}
\begin{definition}
	Let $P,Q:X\to [0,1]$ be two probability distributions with $\supp(Q)\subset \supp(P)$. The relative entropy, also known as the Kullback-Leibler divergence, is defined as
	\begin{equation}
		D(Q||P) = 
		\sum_{\supp(Q)} Q(i)\log \left( \frac{Q(i)}{P(i)} \right).
	\end{equation}
	Note that the relative entropy is always non-negative.
\end{definition}
For $n\in\N$, $P^{\otimes n}:X^n\to [0,1]$ is the $n$'th product distribution given by $P^{\otimes n}(I)=\prod_{j=1}^n P(I_j)$. We wish to study majorization of $P^{\otimes n}$ by $Q^{\otimes n}$ for large $n$. To this end, given a value $v$, we are interested in the size of the set of multiindicies $I$, such that $P^{\otimes n}(I)\ge v$ and the sum of these probabilities. In order to asymptotically compare these for different probability distributions, it is useful to let $v$ depend exponentially on $n$ and look at asymptotic growth rates.
\begin{definition}\label{def1}
	For $V\in \left[\log P(d),\log P(1)\right]$ let
	%\begin{equation}
	%M(V)=\lim_{n\to \infty} \frac{1}{n}\log\sum_{\stackrel{I\in [d]^n}{p_I\ge 2^{nV}}} p_I.
	%\end{equation}
	\begin{equation}
	m_n^P(V)=\sum_{\stackrel{I\in [d]^n}{P^{\otimes n}(I)\ge 2^{nV}}} P^{\otimes n}(I),
	\end{equation}
	\begin{equation}
	m_{n*}^P(V)=\sum_{\stackrel{I\in [d]^n}{P^{\otimes n}(I)\le 2^{nV}}} P^{\otimes n}(I),
	\end{equation}
	\begin{equation}
	s_n^P(V)=\left| \{I\in [d]^n\lv P^{\otimes n}(I)\ge 2^{nV} \}\right|,
	\end{equation}
	\begin{equation}
	s_{n*}^P(V)=\left| \{I\in [d]^n\lv P^{\otimes n}(I)\le 2^{nV} \}\right|.
	\end{equation}
	We define asymptotic exponents of these functions as follows:
	\begin{equation}\label{eq10}
	M^P(V)=\lim_{n\to \infty}\frac{1}{n}\log m_n^P(V),
	\end{equation}
	\begin{equation}\label{eq11}
	M_*^P(V)=\lim_{n\to \infty}\frac{1}{n}\log m_{n*}^P(V),
	\end{equation}
	\begin{equation}\label{eq12}
	S^P(V)=\lim_{n\to \infty}\frac{1}{n}\log s_n^P(V),
	\end{equation}
	\begin{equation}\label{eq13}
	S_*^P(V)=\lim_{n\to \infty}\frac{1}{n}\log s_{n*}^P(V).
	\end{equation}
\end{definition}
It is not immediately clear that the limits describing $M^P,M^P_*,S^P$ and $S_*^P$ are well defined. This will follow from Proposition \ref{HayashiLemma}. The letters chosen, stand for value, mass and size. $V,M^P$ and $S^P$ might be called the value, mass and size exponents, respectively. $M_*^P$ and $S_*^P$ might then be called the converse mass and size exponents. For the purpose of proving Proposition \ref{HayashiLemma}, we need Lemma \ref{concavityLemma}. It should be said that Proposition \ref{HayashiLemma} has been extracted from \cite{MR1960075}, and should merely be viewed as a concise summary and slight extension of some of the tools presented in that paper.
\begin{lemma}\label{concavityLemma}
	Let $X\subset \R^n$ be a compact, convex set. Let $g:X\to \R$ be continuous and $h:X\to \R$ be continuous and strictly concave. Suppose $h$ takes its maximum value at $x_2\in X$. 
	\\
	If $g$ takes its minimum value at $x_1\in X$, then
	\begin{equation}
	y\mapsto \max_{x:g(x)=y}h(x)\quad y\in [g(x_1),g(x_2)]
	\end{equation}
	is strictly monotone increasing.
	\\
	If $g$ takes its maximum value at $x_1\in X$. Then
	\begin{equation}
	y\mapsto \max_{x:g(x)=y}h(x)\quad y\in [g(x_2),g(x_1)]
	\end{equation}
	is strictly monotone decreasing.
\end{lemma}
\begin{proof}
	Assume that $g$ takes its minimum value at $x_1$. Let $g(x_1)\le y'<y''\le g(x_2)$. Let $x'\in g^{-1}\left(y'\right)$ such that $\max_{x:g(x)=y'}h(x)=h(x')$. By continuity of $g$ we may find $x''$ on the line segment between $x'$ and $x_2$, such that $g(x'')=y''$. That is
	\begin{equation}
	x''=\lambda x_2+(1-\lambda)x'
	\end{equation}
	for some $\lambda\in (0,1]$. Since $h$ is strictly concave
	\begin{equation}
	h(x'') \ge \lambda h(x_2)+(1-\lambda)h(x')>h(x').
	\end{equation}
	So
	\begin{equation}
	\max_{x:g(x)=y'}h(x)=h(x') < h(x'') \le \max_{x:g(x)=y''}h(x).
	\end{equation}
	The second part of the lemma follows from the first by replacing $g$ with $-g$.
\end{proof}
Given a probability distribution $P$ with support $[d]$, we let
\begin{equation}
	F_P(\alpha) = \log \sum P(i)^\alpha.
\end{equation}
In order to make things simpler, we shall only consider $F_P$ for probability distributions that are non-uniform (such that $F_P$ is strictly convex) and ordered non-increasingly (such that we may simply write $P(1)$ instead of $\max_{x\in X} P(x)$ and $P(d)$ instead of $\min_{x\in X} P(x)$).
\\\\
The function $F_P$ will be central to the rest of the paper. Note that
\begin{equation}
F_P'(\alpha) = \frac{\sum P(i)^\alpha \log P(i)}{\sum P(i)^\alpha}
\end{equation}
is negative and monotone increasing $F_P':\R\to \left( \log P(d), \log P(1) \right)$. We shall define 
\begin{equation} 
F_P'(\infty)=\lim_{\alpha\to \infty} F_P'(\alpha)=\log P(1)
\end{equation} 
and 
\begin{equation} 
F_P'(-\infty)=\lim_{\alpha\to -\infty} F_P'(\alpha)=\log P(d).
\end{equation} 
$F_P$ is decreasing and strictly convex. Two important values to keep in mind are
\begin{equation}
\begin{split}
	F(0)&= H_0(P) = \log d,
	\\
	F(1)&= 0.
\end{split}
\end{equation}
Also note the following bijections
\begin{equation}
	F_P'
	\begin{cases}
	[-\infty,0]\quad & \longleftrightarrow \quad\left[\log P(d), \frac{\sum_i\log P(i)}{d} \right], \\
	\left[0,1\right]\quad & \longleftrightarrow \quad\left[\frac{\sum_i\log P(i)}{d}, -H(P) \right], \\
	\left[1,\infty \right]\quad & \longleftrightarrow \quad\left[-H(P), \log P(1) \right].
	\end{cases}
\end{equation}
We are now ready to give explicit formulas for the exponent functions (\ref{eq10}),(\ref{eq11}),(\ref{eq12}),(\ref{eq13}).
\begin{proposition}\label{HayashiLemma}
	Let $P$ be a non-uniform probability distribution with $\supp(P)=[d]$ which is ordered non-increasingly. For $V\in \left[ -\log P(d),\log P(1)\right]$ let $\alpha_V\in [-\infty,\infty]$ be the unique solution to $F_P'(\alpha)=V$, then
	\begin{equation}\label{equation16}
	M^P(V) =
	\left\{
	\begin{array}{rl}
	0 & \text{if } V\in \left[-\log P(d),-H(P)\right],\\
	F_P(\alpha_V)+(1-\alpha_V)F_P'(\alpha_V) & \text{if } V \in \left[-H(P),\log P(1)\right].
	\end{array} \right.
	\end{equation}
	\begin{equation}\label{equation17}
	M_*^P(V) =
	\left\{
	\begin{array}{rl}
	F_P(\alpha_V)+(1-\alpha_V)F_P'(\alpha_V) & \text{if } V\in \left[-\log P(d),-H(P)\right],\\
	0 & \text{if } V \in \left[-H(P),\log P(1)\right].
	\end{array} \right.
	\end{equation}
	\begin{equation}\label{equation18}
	S^P(V) =
	\left\{
	\begin{array}{rl}
	\log d & \text{if } V\in \left[-\log P(d),\frac{\sum \log P(i)}{d}\right],\\
	F_P(\alpha_V)-\alpha_V F_P'(\alpha_V) & \text{if } V \in \left[\frac{\sum \log P(i)}{d},\log P(1)\right].
	\end{array} \right.
	\end{equation}
	\begin{equation}\label{equation19}
	S_*^P(V) =
	\left\{
	\begin{array}{rl}
	F_P(\alpha_V)-\alpha_V F_P'(\alpha_V) & \text{if } V\in \left[-\log P(d),\frac{\sum \log P(i)}{d}\right],\\
	\log d & \text{if } V \in \left[\frac{\sum \log P(i)}{d},\log P(1)\right].
	\end{array} \right.
	\end{equation}
	Whenever $\alpha_V=\pm\infty$ the above formulas are to be interpreted as the limit $\alpha\to \pm\infty$.
\end{proposition}
\begin{proof}
	Let $\cP([d])$ be the set of probability distributions on $[d]$.
	The map $h:Q\mapsto H(Q)$ is concave on $\cP([d])$ and takes its maximum value at the uniform distribution, where $-H(Q)-D(Q||P)=\frac{\sum\log P(i)}{d}$. The map $g:Q\to -H(Q)-D(Q||P)$ has maximim value $\log P(1)$ and minimum value $\log P(d)$. According to Lemma \ref{concavityLemma},
	\begin{equation}\label{equation13}
	V\mapsto \max_{-H(Q)-D(Q||P) = V} H(Q)
	\end{equation}
	is strictly monotone decreasing on $\left[\frac{\sum\log P(i)}{d},\log P(1)\right]$ and strictly monotone increasing on $\left[ \log P(d) ,\frac{\sum\log P(i)}{d}\right]$.
	\\\\
	Similarly, according to Lemma \ref{concavityLemma},
	\begin{equation}\label{equation14}
	V\mapsto \max_{-H(Q)-D(Q||P) = V} -D(Q||P)
	\end{equation}
	is strictly monotone increasing on $\left[\log P(d),-H(P)\right]$ and strictly monotone decreasing on $\left[-H(P),\log P(1)\right]$. For each $V\in \left[ \log P(d),\log P(1) \right]$ we wish to find the probability distribution, $Q$, that solves the maximization problems in (\ref{equation13}) and (\ref{equation14}).
	\\\\
	Given $V\in \left( \log P(d), \log P(1) \right) $, let $\alpha$ be the solution to $F_P'(\alpha)=V$ and consider the distribution $P_\alpha(i)=\frac{P(i)^\alpha}{\sum_j P(j)^\alpha}$. Note that $-H(P_\alpha)-D(P_\alpha||P)=F_P'(\alpha)=V$. We prove that $P_\alpha$ solves the above optimization problems. Let $V\in (\log P(d),\log P(1))$ be given and choose $\alpha\in (-\infty,\infty)$ such that $-H(P_\alpha)-D(P_\alpha||P)=V$. Let $Q\in \cP([d])$ be such that also $-H(Q)-D(Q||P)=V$. We need to show that $H(P_\alpha)\ge H(Q)$.
	\begin{equation}
	\begin{split}
	&\ \frac{1}{1-\alpha} \bigg[ D(Q||P_\alpha) - H(P_\alpha)+ H(Q)\bigg]
	\\
	=&\ \frac{1}{1-\alpha}\bigg[D(Q||P_\alpha) + D(P_\alpha|| P) - D(Q||P)\bigg]
	\\
	=&\
	\frac{1}{1-\alpha}\bigg[\sum_i -Q(i)\log\frac{P(i)^\alpha}{\sum_j P(j)^\alpha}+\frac{P(i)^\alpha}{\sum_j P(j)^\alpha}\log \frac{P(i)^\alpha}{\sum_j P(j)^\alpha}-\frac{P(i)^\alpha}{\sum_j P(j)^\alpha}\log P(i)+Q(i)\log P(i) \bigg]
	\\
	=&\
	\frac{1}{1-\alpha}\bigg[(1-\alpha)\sum_i \Big( Q(i) - \frac{P(i)^\alpha}{\sum_j P(j)^\alpha}\Big)\log P(i) \bigg]
	\\ =&\ 
	\sum_i \Big( Q(i) - \frac{P(i)^\alpha}{\sum_j P(j)^\alpha}\Big)\log P(i) = H(P_\alpha) + D(P_\alpha || P)-H(Q) - D(Q || P) = 0.
	\end{split}
	\end{equation}
	So 
	\begin{equation}
	H(P_\alpha) - H(Q) = D(Q||P_\alpha) \ge 0,
	\end{equation}
	which proves that $P_\alpha$ solves the optimization problems with the values
	\begin{equation}
	\begin{split}
	H(P_\alpha) &= -\sum_i \frac{P(i)^\alpha}{\sum_j P(j)^\alpha}\log \frac{P(i)^\alpha}{\sum_j P(j)^\alpha} = \log \sum P(i)^\alpha - \alpha\frac{\sum P(i)^\alpha \log P(i)}{\sum P(i)^\alpha}   
	\\
	&= F_P(\alpha) -\alpha F_P'(\alpha),
	\\\\
	-D(P_\alpha|| P)&=H(P_\alpha)-\Big(H(P_\alpha)+D(P_\alpha||P)\Big) 
	= H(P_\alpha) + F_P'(\alpha)
	\\
	&=
	F_P(\alpha) +(1-\alpha) F_P'(\alpha).
	\end{split}
	\end{equation}
	By standard type class arguments (approximating $P_\alpha$ by type classes while using \cite[Lemma 4 and Lemma 5]{MR1960075} and the fact that the number of type classes grows only polynomially) and monotonicity of the maps (\ref{equation13}) and (\ref{equation14}) we get the following:
	\\
	When $V\in \left[ \frac{\sum \log P(i)}{d},\log P(1) \right)$
	\begin{equation}
	\begin{split}
	S(V)&=\lim_{n\to \infty}\frac{1}{n}\log\Big| \{ I\in [d]^n | p_I\ge 2^{nV} \}\Big| = \max_{-H(Q)-D(Q|| P)\ge V} H(Q)
	\\
	&= \max_{-H(Q)-D(Q|| P)= V} H(Q)=H(P_\alpha)= F_P(\alpha)-\alpha F_P'(\alpha).
	\end{split}
	\end{equation}
	When $V\in \left( \log P(d), \frac{\sum \log P(i)}{d}\right]$
	\begin{equation}
	\begin{split}
	S_*(V)&=\lim_{n\to \infty}\frac{1}{n}\log\Big| \{ I\in [d]^n | p_I\le 2^{nV} \}\Big| = \max_{-H(Q)-D(Q|| P)\le V} H(Q)
	\\
	&= \max_{-H(Q)-D(Q|| P)= V} H(Q)=H(P_\alpha)= F_P(\alpha)-\alpha F_P'(\alpha).
	\end{split}
	\end{equation}
	When $V\in \left(-\log P(d),-H(P)\right]$
	\begin{equation}
	\begin{split}
	M_*(V)&= \lim_{n\to \infty}\frac{1}{n}\log \sum_{\stackrel{I\in [d]^n}{p_I\le 2^{nV}}} p_I= \max_{-H(Q)-D(Q|| P)\le V} -D(Q|| P)
	\\
	&= \max_{-H(Q)-D(Q|| P)= V} -D(Q|| P)=-D(P_\alpha||P)= F_P(\alpha)+(1-\alpha) F_P'(\alpha).
	\end{split}
	\end{equation}
	When $V\in \left[-H(P),\log P(1)\right)$
	\begin{equation}
	\begin{split}
	M(V)&= \lim_{n\to \infty}\frac{1}{n}\log \sum_{\stackrel{I\in [d]^n}{p_I\ge 2^{nV}}} p_I= \max_{-H(Q)-D(Q|| P)\ge V} -D(Q|| P)
	\\
	&= \max_{-H(Q)-D(Q|| P)= V} -D(Q|| P)=-D(P_\alpha||P)= F_P(\alpha)+(1-\alpha) F_P'(\alpha).
	\end{split}
	\end{equation}
	We may take $\alpha$ to $-\infty$ or $\infty$ and get the results at the boundary.
\end{proof}
\begin{remark}\label{remark1}
	Define $\overline{m_n^P}(V)$ on $\left[ \log P(d),\log P(1) \right]$ to be equal to  $m_n^P(V)$ at the endpoints, but for $V\in \left( \log P(d),\log P(1) \right)$ we use a strict inequality and define
	\begin{equation}
		\overline{m_n^P}(V)=\sum_{\stackrel{I\in [d]^n}{P^{\otimes n}(I)> 2^{nV}}} P^{\otimes n} (I).
	\end{equation}
	Define $\overline{m_{n*}^P},\overline{s_{n}^P}$ and $\overline{s_{n*}^P}$ similarly. By continuity of $M^P,M_*^P,S^P$ and $S_*^P$ one sees that we could replace $m_n^P,m_{n*}^P,s_n^P,s_{n*}^P$ in equations (\ref{eq10}),(\ref{eq11}),(\ref{eq12}),(\ref{eq13}) with respectively $\overline{m_{n}^P},\overline{m_{n*}^P},\overline{s_{n}^P},\overline{s_{n*}^P}$, without the limit changing. Furtermore since all functions are monotone and the limit functions are monotone, continuous and bounded, the convergences are all uniform. This will be important later.
\end{remark}
	A few nice values to keep in mind for $S^P,M^P$ and $M_*^P$ are the following
	\begin{equation}
	\begin{split}
		&M^P(-H(P))=0
		\\
		&M^P(\log P(1)) = \log P(1)+\log \left| \left\{i\in [d]\ |\ P(i)=P(1) \right\}\right|
		\\
		&M_*^P(-H(P))=0
		\\
		&S^P(\log P(1))=\log \left| \left\{i\in [d]\ |\  P(i)=P(1) \right\}\right|
		\\
		&S^P(-H(P)) = H(P)
		\\
		&S^P\left(\sum \frac{\log P(i)}{d} \right) = \log d=H_0(P).
	\end{split}
	\end{equation}

\section{A sufficient and almost necessary condition for asymptotic majorization}
\begin{lemma}\label{convexityLemma}
	Let $F_1, F_2:[0,1]\to \R$ be differentiable with $F_2$ convex, such that $F_1(x) > F_2(x)$ for all $x\in [0,1]$. Let $\epsilon < \min_{x\in [0,1]} F_1(x)-F_2(x)$. Then for all $x,y\in [0,1]$
	\begin{equation}\label{eq18}
	F_1(x)-xF_1'(x) \le F_2(y)-yF_2'(y)+\epsilon \implies
	F_1(x)+(1-x)F_1'(x) > F_2(y)+(1-y)F_2'(y) +\epsilon .
	\end{equation}
\end{lemma}
\begin{proof}
	We prove the assertion by contraposition.
	Fix $x,y\in [0,1]$ and assume that 
	\begin{equation}
	F_1(x)+(1-x)F_1'(x) \le F_2(y)+(1-y)F_2'(y)+\epsilon.
	\end{equation} 
	Consider the function 
	\begin{equation}
	g(t)=F_2(y)+(t-y)F_2'(y)-F_1(x)-(t-x)F_1'(x).
	\end{equation} 
	By assumption $g(1)\ge-\epsilon$ and by convexity of $F_2$ 
	\begin{equation}
	g(x) = F_2(y)+(x-y)F_2'(y)-F_1(x)\le F_2(x)-F_1(x)<-\epsilon.
	\end{equation}
	Since $g$ is linear with $g(1)\ge -\epsilon$ and $g(x)<-\epsilon$, we must have $g(0)<-\epsilon$ which is equivalent to
	\begin{equation}
	F_1(x)-xF_1'(x) > F_2(y)-yF_2'(y)+\epsilon.
	\end{equation}
\end{proof}
\begin{lemma}\label{convexityLemma2}
	Let $F_1, F_2:[1,\infty)\to \R$ be continuously differentiable, strictly decreasing and convex, such that $F_1(x) > F_2(x)$ for all $x\in [1,\infty)$. Assume further that $\lim_{x\to \infty} F_1'(x) > \lim_{x\to \infty} F_2'(x)$. Then for sufficiently small $\epsilon>0$ and all $x,y\in [1,\infty)$
	\begin{equation}\label{eq7}
	F_1(x)-xF_1'(x) +\epsilon \ge F_2(y)-yF_2'(y)
	\implies 
	F_1(x)+(1-x)F_1'(x) > F_2(y)+(1-y)F_2'(y) +\epsilon.
	\end{equation}
\end{lemma}
\begin{proof}
	Again we prove the assertion by contraposition.
	Choosing 
	\begin{equation}
		\epsilon< \frac{\lim_{x\to \infty} F_1'(x)-\lim_{x\to \infty} F_2'(x)}{2},
	\end{equation} 
	we have $F_1'(x)>F_2'(y)+2\epsilon$ for all sufficiently large $x$ and all $y\in [1,\infty)$, which implies (\ref{eq7}). For $x=1$, the left hand side of (\ref{eq7}) is never true, so (\ref{eq7}) holds. We thus only need to check (\ref{eq7}) for $x\in (1,R]$, where $R>1$ is some large number.
	\\\\
	Choose epsilon such that 
	\begin{equation}
		\epsilon + R^2\epsilon < \min_{z\in [1,R]}  \frac{F_1(z)-F_2(z)}{R}.
	\end{equation} 
	Let $x\in (1,R],y\in [1,\infty)$, 
	and assume that $F_1(x)+(1-x)F_1'(x) \le F_2(y)+(1-y)F_2'(y)+\epsilon$. Consider the function 
	\begin{equation}
		g(t)=F_2(y)+(t-y)F_2'(y)-F_1(x)-(t-x)F_1'(x).
	\end{equation} 
	By assumption $g(1)\ge -\epsilon$ and by convexity of $F_2$ 
	\begin{equation}
	g(x) = F_2(y)+(x-y)F_2'(y)-F_1(x)\le F_2(x)-F_1(x) \le \max_{z\in [1,R]} F_2(z)-F_1(z).
	\end{equation}
	Since $g$ is linear with $g(1)\ge -\epsilon$ and $g(x)<\max_{z\in [1,R]} F_2(z)-F_1(z)$ we have 
	\begin{equation}
	\begin{split}
		g(0)
		&=\frac{xg(1)-g(x)}{x-1} 
		\ge \frac{-x\epsilon-\max_{z\in [1,R]} F_2(z)-F_1(z)}{x-1}
		\ge \frac{-R\epsilon+\min_{z\in [1,R]} F_1(z)-F_2(z)}{x-1} 
		\\
		&= \frac{\min_{z\in [1,R]} \frac{F_1(z)-F_2(z)}{R}-\epsilon}{R(x-1)}
		>
		\frac{\min_{z\in [1,R]} \frac{F_1(z)-F_2(z)}{R}-\epsilon}{R^2}> \epsilon.
	\end{split}
	\end{equation}
	$g(0)>\epsilon$ is equivalent to the left-hand-side of (\ref{eq7})
\end{proof}
We now combine Proposition \ref{HayashiLemma} with (\ref{eq18}) and (\ref{eq7}). Firstly, by considering the formulas (\ref{equation17}) and (\ref{equation18}), and applying Lemma \ref{convexityLemma} to $F_1=F_P$ and $F_2=F_Q$ one gets
\begin{lemma}\label{lemma1}
	Let $P$ and $Q$ be non-uniform probobability distributions with
	\begin{equation}
		\min_{\alpha\in [0,1]} \frac{H_\alpha(P)}{H_\alpha(Q)} > 1.
	\end{equation}
	For sufficiently small $\epsilon>0$ and all $V\in \left[ \frac{\sum \log P(i)}{d_1},-H(P)\right]$ and $W\in \left[ \frac{\sum \log Q(i)}{d_2},-H(Q)\right]$
	\begin{equation}
		S^P(V)\le S^Q(W)+\epsilon \implies M_*^P(V)\ge M_*^Q(W)+\epsilon.
	\end{equation}
\end{lemma}
Secondly, by considering the formulas (\ref{equation16}) and (\ref{equation18}), and applying Lemma \ref{convexityLemma2} to $F_1=F_Q$ and $F_2=F_P$ one gets
\begin{lemma}\label{lemma2}
	Let $P$ and $Q$ be non-uniform probobability distributions with
	\begin{equation}
		\min_{\alpha\in [1,\infty]} \frac{H_\alpha(P)}{H_\alpha(Q)} > 1.
	\end{equation}
	For sufficiently small $\epsilon>0$ and all $V\in \left[ -H(P), \log P(1)\right]$ and $W\in \left[ -H(Q),-\log Q(1) \right]$
	\begin{equation}
	S^P(V)\le S^Q(W)+\epsilon \implies M_*^P(V)+\epsilon \le M_*^Q(W).
	\end{equation}
\end{lemma}
\begin{proposition}\label{prop1}
	Let $P=P^\downarrow:[d_1]\to [0,1]$ and $Q=Q^\downarrow:[d_2]\to [0,1]$ be non-uniform probobability distributions with
	\begin{equation}
	\min_{\alpha\in [0,1]} \frac{H_\alpha(P)}{H_\alpha(Q)} > 1.
	\end{equation}
	Let $V^*$ be such that $S^P(V^*)\in \left(H(Q),H(P)\right)$.
	Then for all sufficiently large $n$, and all $N$ such that $V=\frac{1}{n}\log(P^{\otimes n\downarrow}(N))\in  \left[\log P(d_1),V^*\right]$
	\begin{equation}\label{eq6}
		\sum_{i=1}^{N-1} P^{\otimes n\downarrow}(i)\le \sum_{i=1}^{N-1} Q^{\otimes n\downarrow}(i).
	\end{equation}
\end{proposition}
\begin{proof}
	Let $\epsilon>0$ be small enough that Lemma \ref{lemma1} applies. Assuming $\epsilon< H_0(P)-H_0(Q)$, we may let $V_1$ be such that $S_P(V_1)=H_0(Q)+\epsilon$. Now let $n$ be large enough that for both $P$ and $Q$ and all $V$ and $W$ (\ref{eq10}),(\ref{eq11}),(\ref{eq12}),(\ref{eq13}) are good approximations, and also good approximations when replaced by the alternative versions in Remark \ref{remark1} (this may be done since the convergences are uniform). Note that $P^{\otimes n\downarrow}(N)=2^{nV}$ such that $N\ge \overline{s_n^P}(V)$.
	We split into three cases. 
	\\\\
	
	First assume that $V\in \left[ \log P(d_1),V_1 \right]$. Then
	\begin{equation}
		\frac{1}{n} \log  N \ge \frac{1}{n}\log \overline{s_n^P}(V) \ge \frac{1}{n}\log \overline{s_n^P}(V_1) \simeq S^P(V_1) = H_0(Q)+\epsilon>H_0(Q),
	\end{equation}
	which implies $N> 2^{nH_0(Q)}$, so
	\begin{equation}
		\sum_{i=1}^{N-1} Q^{\otimes n\downarrow}(i) = 1
	\end{equation}
	and (\ref{eq6}) holds trivially.
	\\\\
	Assume now that $V\in \left[ V_1,-H(P) \right]$. Let $W\in \left[ \frac{\sum \log Q(i)}{d_2},-H(Q) \right]$ be such that $S^Q(W)+\epsilon=S^P(V)$. This is possible by the definition of $V_1$ and by assuming that $\epsilon <H(P)-H(Q)$. Then
	\begin{equation}
		S^P(V)\le S^Q(W)+\epsilon,
	\end{equation}
	which by Lemma \ref{lemma1} implies
	\begin{equation}
		M_*^P(V) \ge M_*^Q(W)+\epsilon.
	\end{equation}
	And
	\begin{equation}\label{eq14}
	\frac{1}{n} \log  N \ge \frac{1}{n} \log \overline{s_n^P}(V) \simeq
	S^P(V) = S^Q(W)+\epsilon > S^Q(W) \simeq \frac{1}{n}\log {s_n^Q}(W),
	\end{equation}
	which implies $N>{s_n^Q}(W)$. Therefore
	\begin{equation}\label{eq15}
	\begin{split}
		&\frac{1}{n} \log\sum_{i=N}^{\infty} P^{\otimes n\downarrow}(i)
		\ge		
		\frac{1}{n} \log \overline{m_{n*}^P}(V)
		\simeq 
		M_*^P(V)
		\ge M_*^Q(W)+\epsilon 
		\\
		&>
		M_*^Q(W)
		\simeq 
		\frac{1}{n} \log m_{n*}^Q(W)
		=
		\frac{1}{n}\log \sum_{\stackrel{I\in [d_2]^n}{Q^{\otimes n}(I)\le 2^{nV}}} Q^{\otimes n}(I)
		\\
		&\ge
		\frac{1}{n} \log \sum_{i={s_n^Q}(W)}^{\infty} Q^{\otimes n\downarrow}(i)
		\ge
		\frac{1}{n} \log \sum_{i=N}^{\infty} Q^{\otimes n\downarrow}(i),
	\end{split}
	\end{equation}
	which implies (\ref{eq6}).
	\\\\
	Finally, assume that $V\in [-H(P),V^*]$. Let $W$ be such that $S^Q(W)\in \left( H(Q),S^P(V^*) \right)$. Again, as in (\ref{eq14}), $N>{s_n^Q}(W)$, and since $M_*^P(V)=0>M_*^Q(W)$, we conclude as in (\ref{eq15}) that 
	\begin{equation}
	\frac{1}{n} \log\sum_{i=N}^{\infty} P^{\otimes n\downarrow}(i)\ge \sum_{i=N}^{\infty} Q^{\otimes n\downarrow}(i).
	\end{equation}
\end{proof}
\begin{proposition}\label{prop2}
	Let $P$ and $Q$ be non-uniform probobability distributions with
	\begin{equation}
	\min_{\alpha\in [1,\infty]} \frac{H_\alpha(P)}{H_\alpha(Q)} > 1.
	\end{equation}
	Let $V^*$ be such that $S^P(V^*)\in \left(H(Q),H(P)\right)$.
	Then for all sufficiently large $n$, and all $N$ such that $V=\frac{1}{n}\log(P^{\otimes n\downarrow}(N))\in  \left[ V^*, \log P(1) \right]$
	\begin{equation}\label{eq16}
	\sum_{i=1}^{N} P^{\otimes n\downarrow}(i)\le \sum_{i=1}^{N} Q^{\otimes n\downarrow}(i)
	\end{equation}
\end{proposition}
\begin{proof}
	Like in the proof of Proposition \ref{prop1} we let $\epsilon>0$ be small enough that Lemma \ref{lemma2} applies.
	We split into three cases.
	Letting $\epsilon>0$ be sufficiently small we may let $W^*\in \left( \log P(1),\log Q(1) \right)$ be the solution to $S^Q(W^*)=S^Q(\log Q(1))+\epsilon$.
	\\\\
	Firstly we assume that $S^P(V) \le S^Q(W^*)$, then 
	\begin{equation}
		\frac{1}{n}\log N \le \frac{1}{n} \log s_n^P(V)\simeq S^P(V) \le S^Q(W^*)<S^Q(\log P(1))\sim \frac{1}{n}\log s_n^Q(\log P(1))
	\end{equation}
	showing that $N\le s_n^Q(\log P(1))$ which implies that $Q^{\otimes N\downarrow}(i)\ge \log P(1)$ for all $i\in [N]$. So
	\begin{equation}
	\begin{split}
		\sum_{i=1}^{N} Q^{\otimes n\downarrow}(i)\ge N\log P(1) \ge  \sum_{i=1}^{N} P^{\otimes n\downarrow}(i)
	\end{split}
	\end{equation}
	Secondly we assume that $S^P(V)\in \left[ S^Q(W^*) , H(Q) \right]$. Let $W\in \left[ -H(Q),\log Q(1) \right]$ be such that $S^Q(W)+\epsilon = S^P(V)$, which is possible by the choice of $W^*$. By Lemma \ref{lemma2}
	\begin{equation}
		M^P(V)+\epsilon \le M^Q(W).
	\end{equation}
	And
	\begin{equation}
		\frac{1}{n} \log N \ge \frac{1}{n} \log \overline{s_n^P}(V)\simeq S^P(V) = S^Q(W)+\epsilon/2 > S^Q(W) \simeq \frac{1}{n} \log s_n^Q(W),
	\end{equation}
	showing that $N> s_n^Q(W)$.
	\begin{equation}
	\begin{split}
		&\frac{1}{n}\log \sum_{i=1}^N P^{\otimes n\downarrow}(i)
		\le 
		\frac{1}{n}\log m_n^P(V) \simeq M^P(V) 
		\\
		&< 
		M^P(V)+\epsilon
		\le M^Q(W)\simeq \frac{1}{n} \log m_n^Q(W)
		\\
		&
		=
		\frac{1}{n} \log \sum_{i=1}^{s_n^Q(W)} Q^{\otimes n\downarrow}(i)
		\le
		\frac{1}{n}\log \sum_{i=1}^{N} Q^{\otimes n\downarrow}(i).
	\end{split}
	\end{equation}
	Finally assume that $S^P(V)\in\left[ H(Q),S^P(V^*) \right]$. Let $W > -H(Q)$ be such that $M^Q(W)>M^P(V^*)$. Then
	\begin{equation}
	\frac{1}{n} \log N \ge \frac{1}{n}\log \overline{s_n^P}(V) \simeq S^P(V) \ge H(Q) > S^Q(W) \simeq \frac{1}{n}\log s_n^Q(W),
	\end{equation} 
	showing that $N>s_n^Q(W)$.
	\begin{equation}
	\begin{split}
	&\frac{1}{n} \log\sum_{i=1}^{N} P^{\otimes n\downarrow}(i)
	\le		
	\frac{1}{n} \log {m_{n}^P}(V)
	\simeq 
	M^P(V)
	\le M^P(V^*)< M^Q(W) 
	\\
	&\simeq 
	\frac{1}{n} \log m_{n}^Q(W)
	=
	\frac{1}{n}\log \sum_{\stackrel{I\in [d_2]^n}{Q^{\otimes n}(I)\ge 2^{nW}}} Q^{\otimes n}(I)
	\\
	&=
	\frac{1}{n} \log \sum_{i=1}^{s_n^Q(W)} Q^{\otimes n\downarrow}(i)
	\le
	\frac{1}{n} \log \sum_{i=1}^{N} Q^{\otimes n\downarrow}(i).
	\end{split}
	\end{equation}'
\end{proof}
So far we have assumed that all probability distributions are non-uniform. This was mainly a matter of convenience. In the following we no longer make this assumption. If $Q$ is the trivial probability distribution (i.e. $\left|\supp(Q)\right|=1$), then $P^{\otimes n} \preceq Q^{\otimes n}$ holds for any $P$ and $n$, so this case is rather uninsteresting.
\begin{proposition}\label{prop3}
	Let $P=P^\downarrow:[d_1]\to \R$ and $Q=Q^\downarrow:[d_2]\to \R$ be two probability distributions with $d_2>1$ and assume that
	\begin{equation}
		\min_{\alpha\in [0,\infty]} \frac{H_\alpha(P)}{H_\alpha(Q)} > 1.
	\end{equation}
	For sufficiently large $n$
	\begin{equation}
		P^{\otimes n} \preceq Q^{\otimes n}
	\end{equation}
\end{proposition}
\begin{proof}
	If $d_1=1$ then ${H_\alpha(P)}=0$ for all $\alpha$, so we may assume that $d_1>1$.
	For small $\delta>0$, let
	\begin{equation}
		P_\delta(i) =
		\begin{cases}
			P(1)+\delta & \text{if } i=1,\\
			P(i) & \text{if } 1<i<d_1,\\
			P(d_1)-\delta & \text{if } i=d_1.
		\end{cases}
	\end{equation}
	\begin{equation}
	Q_\delta(i) =
	\begin{cases}
		Q(1)-\delta & \text{if } i=1,\\
		Q(i) & \text{if } 1<i<d_2,\\
		Q(d_1)+\delta & \text{if } i=d_2.
	\end{cases}
	\end{equation}
	When $\delta$ is sufficiently small
	\begin{equation}
		\min_{\alpha\in [0,\infty]} \frac{H_\alpha(P_\delta)}{H_\alpha(Q_\delta)} > 1.
	\end{equation}
	By applying Propositions \ref{prop1} and \ref{prop2} to $P_\delta$ and $Q_\delta$, we get for large $n$
	\begin{equation}
		P^{\otimes n}\preceq P_\delta^{\otimes n} \preceq Q_\delta^{\otimes n} \preceq Q^{\otimes n}.
	\end{equation}
\end{proof}
We have now established a sufficient condition for asymptotic majorization. In fact this condition is almost necesarry. 
It is well known that for $\alpha\in (0,\infty)$ the $\alpha$-R\'enyi entropy is strictly Schur-concave (this is a consequence of the fact that $p\mapsto \frac{1}{1-\alpha}p^\alpha$ is strictly concave). In other words:
\begin{proposition}\label{prop5}
	Let $P$ and $Q$ be two probability distribution with $P\preceq Q$. Then either
	\begin{equation}
	P^\downarrow=Q^\downarrow
	\end{equation}
	or
	\begin{equation}
	H_\alpha(P)>H_\alpha(Q)\quad \text{for all } \alpha\in (0,\infty)
	\end{equation}
\end{proposition}
Using the fact that $H_\alpha(P^{\otimes n})=nH_\alpha (P)$, we may sum up the contents of Propositions \ref{prop3} and \ref{prop5} as follows: When $P^\downarrow\not = Q^\downarrow$;
\begin{equation}\label{eq24}
	\begin{split}
	\forall \alpha\in \left[0,\infty\right]&:H_{\alpha}(P)>H_{\alpha}(Q)
	\\
	&\Downarrow \ref{prop3}
	\\
	  \exists n\in \N&: P^{\otimes n} \preceq Q^{\otimes n}
	 \\
	 &\Downarrow \ref{prop5}
	 \\
	 \forall \alpha\in (0,\infty)&:H_{\alpha}(P)>H_{\alpha}(Q).
	\end{split}
\end{equation}
\begin{remark}
It is natural to ask if we can make requirements at $0$ and $\infty$ in order to get a biimplication, that is, if we can determine $\exists n\in \N: P^{\otimes n} \preceq Q^{\otimes n}$ entirely from comparing Rényi entropies. It seems that in order to do so, we would have to be more careful with our estimations. The author cautiously conjectures that requiring a weak inequality at $\infty$ is sufficient, and that the requirement of a sharp inequality at $0$ could be replaced by a similar condition regarding the $\alpha$-R\'enyi entropies for negative $\alpha$.
\end{remark}
\begin{definition}
	When $P$ and $Q$ are probability distributions with finite support, we let
	\begin{equation}
		R(P,Q) = \sup \left\{r\in \R_{\ge 0} \big| \text{ for large } n\ P^{\otimes n} \preceq Q^{\otimes \lfloor nr\rfloor} \right\}.
	\end{equation}
\end{definition}
When $Q$ is the trivial probability distribution $R(P,Q) = \infty$.
\begin{theorem}\label{th1}
	Given finitely supported probability distributions $P$ and $Q$, with $Q$ non-trivial,
	\begin{equation}
		R(P,Q) = \min_{\alpha\in [0,\infty]} \frac{H_\alpha(P)}{H_\alpha(Q)}.
	\end{equation}
\end{theorem}
\begin{proof}
	Let $r< \min_{\alpha\in [0,\infty]} \frac{H_\alpha(P)}{H_\alpha(Q)}$. Then for large $n$
	\begin{equation}
			\min_{\alpha\in [0,\infty]}\frac{H_\alpha(P^{\otimes n})}{H_\alpha(Q^{\otimes \lfloor nr\rfloor})}
			=
			\min_{\alpha\in [0,\infty]}
			\frac{n}{\lfloor nr\rfloor}\frac{H_{\alpha}(P)}{H_{\alpha}(Q^)}
			>1.
	\end{equation}
	By Proposition \ref{prop3}, $P^{\otimes n}\preceq Q^{\otimes \lfloor nr\rfloor} $.
	\\\\
	Let $r> \min_{\alpha\in [0,\infty]} \frac{H_\alpha(P)}{H_\alpha(Q)}$ and choose some ${\alpha_r}$ such that $r>\frac{H_{\alpha_r}(P)}{H_{\alpha_r}(Q)}$. Then for large $n$
	\begin{equation}
	\frac{H_{\alpha_r}(P^{\otimes n})}{H_{\alpha_r}(Q^{\otimes \lfloor nr\rfloor})}
	=
	\frac{n}{\lfloor nr\rfloor}\frac{H_{\alpha_r}(P)}{H_{\alpha_r}(Q^)}
	<1.
	\end{equation}
	By Proposition \ref{prop5} $P^{\otimes n}\npreceq Q^{\otimes \lfloor nr\rfloor} $.
\end{proof}
\paragraph{Acknowledgements.}
I acknowledge financial support from the European Research Council (ERC Grant Agreement no. 337603) and VILLUM FONDEN via the QMATH Centre of Excellence (Grant no. 10059).
\\
Furthermore I would like to thank Péter Vrana for stimulating discussions.
\bibliographystyle{ieeetr}
\bibliography{all}

\begin{thebibliography}{1}

\bibitem{PhysRevLett.83.436}
M.~A. Nielsen, ``Conditions for a class of entanglement transformations,'' {\em
  Phys. Rev. Lett.}, vol.~83, pp.~436--439, Jul 1999.

\bibitem{fritz_2017}
T.~Fritz, ``Resource convertibility and ordered commutative monoids,'' {\em
  Mathematical Structures in Computer Science}, vol.~27, no.~6, p.~850–938,
  2017.

\bibitem{Nielsen:2011:QCQ:1972505}
M.~A. Nielsen and I.~L. Chuang, {\em Quantum Computation and Quantum
  Information: 10th Anniversary Edition}.
\newblock New York, NY, USA: Cambridge University Press, 10th~ed., 2011.

\bibitem{2018arXiv180705130K}
A.~K. {Jensen} and P.~{Vrana}, ``{The asymptotic spectrum of LOCC
  transformations},'' {\em arXiv:1807.05130}, July 2018.

\bibitem{MR1960075}
M.~Hayashi, M.~Koashi, K.~Matsumoto, F.~Morikoshi, and A.~Winter, ``Error
  exponents for entanglement concentration,'' {\em J. Phys. A}, vol.~36, no.~2,
  pp.~527--553, 2003.

\end{thebibliography}
\end{document}